\documentclass{article}
\usepackage[english]{babel}
\usepackage{amsthm}
\usepackage{enumerate}
\usepackage{mathtools}
\usepackage{amsfonts}
\usepackage{csquotes}
\usepackage{fullpage}
\usepackage[sorting=nyt,style=numeric,maxbibnames=10,backend=biber]{biblatex}
\usepackage[hidelinks]{hyperref}
\usepackage[noabbrev,capitalise,nameinlink]{cleveref}
\usepackage{tabu}
\usepackage{siunitx}
\usepackage{xfp} 
\usepackage{subcaption}
\usepackage{caption}
\usepackage{microtype}
\usepackage{xspace}
\usepackage{diagbox}
\usepackage{paralist}
\usepackage{nicefrac}

\AtEveryBibitem{
    \clearfield{file}
    \clearfield{month}
    \clearfield{day}
}
\AtEveryBibitem{%
  \ifentrytype{online}{}{%
    \clearfield{urldate}%
    \clearfield{urlyear}%
    \clearfield{urlmonth}%
  }%
}
\AtEveryBibitem{%
  \iffieldundef{doi}{}{%
    \clearfield{issn}%
    \clearfield{isbn}%
  }%
}

\newcommand{\mO}{\mathcal{O}}

\newcommand{\N}{\mathbb{N}}
\newcommand{\R}{\mathbb{R}}
\newcommand{\cro}{\textup{cr}}
\newcommand{\crk}{\textup{cr}_k}

\newcommand{\gcrk}{\overline{\overline{\textup{cr}}}_k}
\newcommand{\gcr}{\overline{\overline{\textup{cr}}}}
\newcommand{\bcrk}{\textup{bkcr}_k}

\newcommand{\col}{\chi}
\newcommand{\ps}{P} 
\newcommand{\hm}{m} 

\newcommand{\lset}[2]{S_{#1}^\ell(#2)}
\newcommand{\rset}[2]{S_{#1}^r(#2)}
\newcommand{\dset}[2]{S_{#1}^d(#2)}

\newcommand{\cbar}{\overline{c}}
\newcommand{\cprime}{c'}

\newcommand{\gcrknumber}{geometric $k$-colored crossing number\xspace}

\newcommand{\gcrkconstant}{geometric $k$-colored crossing constant\xspace}

\newtheorem{theorem}{Theorem}

\newtheorem{claim}{Claim}
\newtheorem{lemma}{Lemma}

\title{On the \gcrknumber of $K_n$}
\author{Benedikt Hahn, Bettina Klinz, and Birgit Vogtenhuber\\[0.2ex]
{\small{{\tt benedikt.hahn@tugraz.at}, {\tt klinz@math.tugraz.at}, {\tt birgit.vogtenhuber@tugraz.at}}
}\\[0.5ex] {{Graz University of Technology, Austria}}}
\date{\hfill}

\begin{document}
\maketitle

\begin{abstract}
  We study the \emph{\gcrknumber} of complete graphs $\gcrk(K_n)$, which is the smallest number of monochromatic crossings in any $k$-edge colored straight-line drawing of $K_n$.
  We substantially improve asymptotic upper bounds on $\gcrk(K_n)$ for $k=2,\ldots, 10$ by developing a procedure for general $k$ that derives $k$-edge colored drawings of $K_n$ for arbitrarily large $n$ from initial drawings with a low number of monochromatic crossings.
  We obtain the latter by heuristic search, employing a \textsc{MAX-$k$-CUT}-formulation of a subproblem in the process.
\end{abstract}

\section{Introduction}
A \emph{drawing $\Gamma$ of a graph $G$} is a representation of $G$ in $\mathbb{R}^2$ where vertices are represented as distinct points,
edges are represented as simple continuous curves connecting their endpoints, and no curve passes through the representation of a vertex.
For simplicity, we assume that no two curves share more than finitely points or a tangent point and that no three edges share a point in their relative interior.
A \emph{crossing} in $\Gamma$ is a point in the relative interior of two curves.
For brevity, we will mostly refer to the elements of $\Gamma$ as vertices and edges.

Crossing minimization for non-planar graphs is of great interest from both a theoretical and a practical point of view.
The \emph{crossing number} $\cro(G)$ of $G$ is the minimum number of crossings $\cro(\Gamma)$ in any drawing $\Gamma$ of~$G$.
A plethora of variants of crossing numbers have been studied; see for example the survey of Schaefer~\cite{schaeferGraphCrossingNumber2024}.
Despite intensive research, various important problems --- such as determining $\cro(K_n)$ --- still remain open for the \enquote{original} crossing number, and the same holds true for many relevant variants.
One such variant and the topic of this work is the \emph{\gcrknumber} $\gcrk(G)$.
It is defined as\\[-1ex]
\begin{equation}
    \label{def:rkccn}
    \gcrk(G) \coloneqq \min_{\Gamma} \min_{G=G_1 \cup \ldots \cup G_k} \sum_{i = 1}^{k} \cro(\Gamma|_{G_i}),
\end{equation}
where $\Gamma$ ranges over all \emph{straight-line drawings} of $G$ (i.e., drawings in which the edges are straight-line segments).
Equivalently, $\gcrk(G)$ is the minimum number of \emph{monochromatic crossings} in any \emph{straight-line $k$-edge-colored drawing of~$G$}.
Straight-line drawings are also called geometric graphs, which motivates the name \emph{\gcrknumber}
\footnote{
We use the notation of Schaefer~\cite{schaeferGraphCrossingNumber2024} for $\gcrk(G)$, but a different name than in previous literature
for the following reasons: We do not write \enquote{geometric $k$-planar crossing number} as in~\cite{pachNoteKplanarCrossing2018} to avoid confusion with the concept of
$k$-planar graphs, and we do not write \enquote{rectilinear $k$-colored crossing number} as in~\cite{aichholzer2ColoredCrossingNumber2019,fabila-monroyNoteKcoloredCrossing2025} to avoid confusion with the related but different \enquote{rectilinear $k$-planar crossing number} and to highlight the relation to geometric thickness.
}.
The \gcrknumber is closely related to \emph{geometric thickness}, which is the minimum $k$ such that $\gcrk(G)=0$.

For $k=1$, $\gcrk(G)$ is the classical \emph{rectilinear crossing number} of a graph (mostly denoted by $\overline{\textup{cr}}(G)$).
Determining $\gcr_1(G)$ is $\exists \R$-hard~\cite{bienstockProvablyHardCrossing1991,deanMathematicalProgrammingFormulation2002} and exact values for $\gcrk(G)$ are known only for few graph classes.
In particular, despite intensive research, $\gcr_1(K_n)$ is still unknown for $n\ge 31$ and there is no candidate for a closed formula.
On the positive side, the \emph{rectilinear crossing constant} $\gcr_1 \coloneqq \lim_{n \to \infty} \gcr_1(K_n)/\binom{n}{4}$ is known to exist (see for example~\cite{richterRelationsCrossingNumbers1997}).
Its bounds meanwhile have been narrowed to
\begin{equation}
    \label{eq:cr1_bounds}
    0.37997 \leq \gcr_1 \leq 0.38045
\end{equation}
(see~\cite{abrego$leqK$EdgesCrossings2011} and the arXiv-version of~\cite{aichholzerOngoingProjectImprove2020}).
The upper bound employs a construction from~\cite{abregoGeometricDrawings$K_n$2007} that generates drawings of $K_n$ with arbitrarily large $n$ and small rectilinear crossing number given an initial drawing $\Gamma$ with few crossings and a so-called \emph{halving matching} of $\Gamma$ (a matching between vertices with incident halving edges).

For fixed $k \geq 2$, the \emph{\gcrkconstant} $\gcr_k \coloneqq \lim_{n \to \infty} \gcrk(K_n)/\binom{n}{4}$ exists as well (with an identical proof as for $\gcr_1$ in~\cite{richterRelationsCrossingNumbers1997}).
For $k=2$, the previously best known bounds on $\gcr_2$ are 
\begin{equation}
    \label{eq:cr2_bounds}
    \frac{1}{33} = 0.\bar{03} \leq \gcr_2 \leq 0.11798016,
\end{equation}
both of which were shown in~\cite{aichholzer2ColoredCrossingNumber2019}.
The upper bound is obtained by a generalized notion of halving matchings and a construction based on the approach from~\cite{abregoGeometricDrawings$K_n$2007}, but is specifically tailored to two colors.

For $k\geq 3$, bounds on $\gcr_k$ are derived by the following: For any $k\geq 1$ and any graph $G$, $\gcrk(G)$ is bounded from below by the \emph{$k$-colored crossing number} $\crk$ which is defined as in \labelcref{def:rkccn} but with $\Gamma$ ranging over all possible drawings of $K_n$ instead of only straight-line drawings.
On the other hand, $\gcrk(G)$ is bounded from above by the \emph{$k$-page book crossing number} $\bcrk(G)$, also defined as in \labelcref{def:rkccn} but with $\Gamma$ restricted to drawings of $G$ with the $n$ vertices in convex position.
From the existence of $\crk:= \crk(K_n)/\binom{n}{4}$ and $\bcrk:= \bcrk(K_n)/\binom{n}{4}$, combined with the best known asymptotic bounds on $\crk(K_n)$ and $\bcrk(K_n)$, we obtain
\begin{equation}
    \label{eq:crk_bounds}
    \frac{3}{29k^2} \leq \crk \leq \gcrk \leq \bcrk \leq \frac{2}{k^2} - \frac{1}{k^3}.
\end{equation}
The lower bound in \labelcref{eq:crk_bounds} stems from~\cite{shavaliBiplanar$k$PlanarCrossing2022} and is via an application of the Crossing Lemma~\cite{ackermanTopologicalGraphsMost2019}.
The upper bound follows independently from two different constructions~\cite{damianiUpperBoundCrossing1994,shahrokhiBookCrossingNumber1996} as noticed in~\cite{deklerkImprovedLowerBounds2013}.

\paragraph*{Our Contribution.}
In this work, we develop a technique to improve the upper bounds of $\gcrk$ for any fixed $k\geq 2$ by generalizing the approach from~\cite{aichholzer2ColoredCrossingNumber2019} to $k\geq 2$ and by improving it for $k=2$.
To this end, we also find provably \emph{optimal matchings} for any $k\geq 2$ (for $k=1$, the matchings from~\cite{abregoGeometricDrawings$K_n$2007} were known to be optimal; the ones used in~\cite{aichholzer2ColoredCrossingNumber2019} were not).
We exemplify our approach for $k = 2, \ldots, 10$ (using heuristic search methods for initial drawings and colorings) and obtain substantially improved upper bounds for $\gcr_2,\ldots \gcr_{10}$.

\paragraph*{Further related work.}
The \gcrknumber first appeared in~\cite{pachNoteKplanarCrossing2018}.
In the literature, the \gcrknumber has also been treated in the context of the \emph{$k$-colored crossing ratio of a drawing $\Gamma$}, that is, the ratio between the $k$-colored crossing number $\crk(\Gamma)$ of $\Gamma$ and $\cro(\Gamma)$.
The authors of~\cite{aichholzer2ColoredCrossingNumber2019} proved that there is some constant $c > 0$ such that for all large enough $n$ and all straight-line drawings $\Gamma$ of $K_n$, $\crk(\Gamma)/\cro(\Gamma) \leq 1/2 - c$.
In~\cite{fabila-monroyNoteKcoloredCrossing2025} this is generalized to any $k$ and to dense graphs.
In~\cite{cabelloNote2coloredRectilinear2024} it is shown that for $n$ points chosen uniformly at random from a unit square, the induced straight-line drawing $\Gamma$ of $K_n$ has $\crk(\Gamma)/\cro(\Gamma) \leq 1/2 - 7/50$ in expectation.

The crossing properties of $\Gamma$ are captured by the \emph{crossing graph} of $\Gamma$, whose vertices are the edges of~$\Gamma$ and in which two vertices are adjacent if they cross.
As the total number of uncolored crossings in $\Gamma$ is fixed, a $k$-edge coloring of $\Gamma$ realizing $\cro_k(\Gamma)$ is equivalent to a $k$-vertex-coloring of the crossing graph of $\Gamma$ that maximizes adjacencies between differently colored vertices, i.e., a \emph{maximum $k$-cut}.
The problem \textsc{MAX-$k$-CUT} is $\mathcal{NP}$-hard and also hard to approximate in general~\cite{gareySimplifiedNPcompleteGraph1976,kannHardnessApproximatingMax1996}.
Moreover, it remains $\mathcal{NP}$-hard for segment intersection graphs~\cite{masudaCrossingMinimizationLinear1990} (and hence for crossing graphs of drawings)~\cite{masudaCrossingMinimizationLinear1990}.
On the other hand, there is a PTAS for \textsc{MAX-$k$-CUT} for dense graphs~\cite{aroraPolynomialTimeApproximation1999} (and hence for crossing graphs of drawings of $K_n$).
\pagebreak

\section{The doubling construction}
We consider straight-line drawings of $K_n$ given by some set of points $\ps \subseteq \R^2$ in general position with a $k$-edge-coloring $\col$.
We denote the number of monochromatic crossings in the resulting drawing by $\cro_k(P;\col)$.

We work with matchings, which match each vertex with an incident edge such that no edge is matched twice.
Formally, a \emph{matching} is a map $\hm:\ps \to \ps$ with $\hm(p) \neq p$, $\hm(\hm(p)) \neq p$ for all $p \in \ps$.
We call $p \hm(p)$ the \emph{matching edge of $p$} and think of it as being oriented away from $p$.
We denote the color of the matching edge of $p$ as $\cbar(p) \coloneqq \col(p \hm(p))$.
For each color $c$, the number of edges incident to $p$ with color $c$ that lie to the left (respectively right) of the line spanned by the matching edge of $p$ is denoted as $\lset{c}{p}$ (respectively $\rset{c}{p}$).

A \emph{$\chi$-halving matching} as defined in~\cite{aichholzer2ColoredCrossingNumber2019} is a matching with the additional property that for each point $p$, a color $c$ with the maximum number of incident edges at $p$ fulfills $\left| \lset{c}{p} - \rset{c}{p} \right| \leq 1$.

Given a point set $P_0$, a $k$-edge-coloring $\col_0$ and matching $\hm_0$, we construct a point set $P_1$ with $|P_1| = 2|P_0|$ together with a $k$-edge-coloring $\col_1$ and a matching $\hm_1$ in a way that can be iteratively repeated to obtain $\col_t$, $P_t$, and $\hm_t$ for any $t \in \N$ with few monochromatic crossings in the following way.

\noindent{\textbf{\emph{Point set:}}}
We replace each point $p \in P_0$ by the two points with distance $\varepsilon$ to $p$ on the line spanned by the edge $p \hm_0(p)$ for a sufficiently small positive $\varepsilon$ (i.e., such that no smaller $\varepsilon$ changes the order type of the point set).
The resulting points are the \emph{children of $p$}.
We denote them by $p_1$ and $p_2$ such that $p_1$ is further from $\hm_0(p)$ than $p_2$.
In turn, $p$ is the \emph{parent} of $p_1$ and $p_2$.
We further denote the left and right child of $\hm_0$ from the perspective of $p$ as $\hm_0(p)_\ell$ and $\hm_0(p)_r$.

\noindent{\textbf{\emph{Coloring and matching:}}}
We choose $\col_1(p_i q_j) = \col_0(pq)$ if $p \neq q$.
Independently for each vertex $p$, we decide on $\cprime(p) \coloneqq \col_1(p_1 p_2)$ and $\hm_1(p_1), \hm_1(p_2)$ and call these choices the \emph{details (at~$p$)}.
We restrict the choice of $\hm_1(p_1)$ and $\hm_1(p_2)$ to the children of $\hm_0(p)$ and $p_1, p_2$ and disallow $\hm_1(p_2) = p_1$, in order to preserve the rough structure of $\hm_0$; see \cref{fig:doubling_one_step} for an example.
We do not enforce any canonical method to choose the details but will later describe how details that optimize the asymptotic number of crossings can be found.
In contrast, the authors of~\cite{aichholzer2ColoredCrossingNumber2019} choose the details at a vertex $p$ according to a case distinction on the color of $\hm_0(p)$ and the values of $\lset{c}{p}$ and $\rset{c}{p}$, which is not always optimal.

\begin{figure}[h]
    \centering
    \begin{subfigure}[b]{0.45\textwidth}
        \centering
        \includegraphics[page=1]{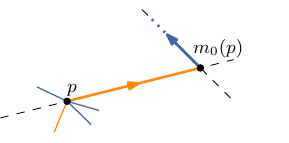}
    \end{subfigure}
	\hspace{1cm}
    \begin{subfigure}[b]{0.45\textwidth}
        \centering
        \includegraphics[page=2]{figures/doubling_one_step.pdf}
    \end{subfigure}        
    \caption{One step in the doubling procedure at a vertex $p$.
            Matching edges are drawn bold and with an arrowhead.
            Dashed lines are the extensions of matching edges along which the vertices are split.}
    \label{fig:doubling_one_step}
\end{figure}

\noindent\textbf{\emph{Iterated application:}}
Given $P_i$, $\col_i$, and $\hm_i$ from iteration $i$, we construct $P_{i+1}$, $\col_{i+1}$, and $\hm_{i+1}$ analogously to the first step.
In particular, each vertex in $P_i$ is a parent of two vertices in $P_{i+1}$.
Calling a vertex in some $P_i$ a \emph{descendant} of $p \in P_0$ if they are transitively related by the parent relation, the descendants of $p$ form an infinite full binary tree $T_p$ rooted at $p$.
We set the left child of $p$ to $p_1$ and the right child to $p_2$.
We denote by $p^i_j$ the vertex on level $i$ (thus in $P_i$) of $T_p$ at position $j \in \{1, \ldots, 2^i\}$ from left to right.
In this notation, $p = p^0_1$, $p_1 = p^1_1$, $p_2 = p^1_2$, and the two children of $p^i_j$ are $p^{i+1}_{2j - 1}$ and $p^{i+1}_{2j}$.

For each descendant $p^i_j \in P_i$ of $p \in P_0$, we choose the details at $p^i_j$ identically to those at $p$:
$\col_{i+1}(p^{i+1}_{2j-1} p^{i+1}_{2j}) = \col_1(p_1 p_2)$ and if $\hm_1(p_1) = \hm_0(p)_\ell$ then $\hm_{i+1}(p^{i+1}_{2j-1}) = \hm_i(p^i_j)_\ell$, if $\hm_1(p_1) = \hm_0(p)_r$ then $\hm_{i+1}(p^{i+1}_{2j-1}) = \hm_i(p^i_j)_r$, and if $\hm_1(p_1) = p_2$ then $\hm_{i+1}(p^{i+1}_{2j-1}) = p^{i+1}_{2j}$ (analogously for $p^{i+1}_{2j}$).
\pagebreak

\section{Analysis}
The following theorem counts the number of crossings after $t$ step of the doubling construction.
The proof can be found in \cref{app:proof_one_step_crossings} and is based on similar counting arguments as in~\cite[Claim 1]{aichholzer2ColoredCrossingNumber2019} and~\cite[Lemma 3]{abregoGeometricDrawings$K_n$2007}.

\begin{theorem}
    \label{thm:t_step_crossings}
	Given a point set $P_0$, a $k$-edge-coloring $\col_0$, a matching $\hm_0$, and details at all vertices of $P_0$, the number of monochromatic crossings after $t$ iterations of the doubling construction is
	\begin{equation}
       \label{fml:t_step_crossings}
	\begin{array}{rcl}
		\crk(P_t; \col_t) &=& 16^t\;\crk(P_0; \col_0) + \sum\limits_{i=0}^{t-1}16^{t-i-1}\left[\binom{2^i|P_0|}{2} - 2^i|P_0|\right] \\
		&&\ +\ 4 \sum\limits_{p \in P_0} \sum\limits_{c = 1}^{k} \sum\limits_{i=0}^{t-1}16^{t-i-1}\sum\limits_{j=1}^{2^i}\left[ \binom{\lset{c}{p^i_j}}{2} + \binom{\rset{c}{p^i_j}}{2} \right] \\
		&&\ +\ 2 \sum\limits_{p \in P_0} \sum\limits_{i=0}^{t-1}16^{t-i-1} \sum\limits_{j=1}^{2^i}\left[\lset{\cbar(p)}{p^i_j} + \rset{\cbar(p)}{p^i_j} \right].
	\end{array}
	\end{equation}
\end{theorem}
   
To determine the asymptotics of \labelcref{fml:t_step_crossings} for $t \to \infty$, we consider the values of $\dset{c}{p^i_j}$.
We reason that there exist offsets $o_1, o_2 \in \{0,1,2\}$ depending only on $p \in P_0$, $c \in \{1, \ldots,k\}$, and $d \in \{\ell,r\}$, such that
$$\dset{c}{p^{i+1}_{2j-1}} = 2 \cdot \dset{c}{p^i_j} + o_1 \qquad \text{ and } \qquad \dset{c}{p^{i+1}_{2j}} = 2 \cdot \dset{c}{p^i_j} + o_2$$
for all $i,j$.
As the doubling procedure behaves identically for all iterations, it is sufficient to make the following arguments for $i=0$.
The factor 2 appears because each edge $pq$, $q \neq \hm_0(p)$ gives rise to two edges incident to $q_1$ and two edges incident to $q_2$ with the same color as $pq$ and on the same side of the respective halving edges.
The offsets $o_1$ and $o_2$ stem from the six edges $p_1 p_2, p_1 q_1, p_1 q_2$ at $p_1$ and $p_2 p_1, p_2 q_1, p_2 q_2$ at $p_2$.
Two of these are the matching edges of $p_1$ and $p_2$ and do not count towards $\dset{c}{p_1}$ or $\dset{c}{p_2}$ for any $c$ and $d$, so $o_1, o_2 \leq 2$.
Further, $o_2 \leq 1$, as no choice for the matching edge of $p_2$ leaves the remaining two edges on the same side.
Finally, if $o_1 = 2$, then $o_2 \neq 0$ as is apparent from a short case distinction.
We show in \cref{app:reccurence_formulas} that the various values of $o_1$ and $o_2$ adhere the following five closed formulas for $\dset{c}{p^i_j}$, which we denote as 
$f_{(o_1,o_2)}(\dset{c}{p},i,j)$.

\begin{table}[h]
	\centering
	{\small{
    {\tabulinesep=.2mm
    \begin{tabu}{c|c|c|c}
        \diagbox[height=1.5em]{$o_2$}{$o_1$} & 0 & 1 & 2\\
        \hline
        0 & $2^i \dset{c}{p}$       & $2^i \dset{c}{p} + 2^i - j$ &  - \\
        1 & $2^i \dset{c}{p} + j-1$ & $2^i \dset{c}{p} + 2^i - 1$ & $2^i \dset{c}{p} + 2^{i+1} - j - 1$ \\
    \end{tabu}}
}}
\end{table}

\noindent Plugging these into \labelcref{fml:t_step_crossings}, we obtain
\begin{equation}
\begin{array}{rcl}
    \label{fml:t_step_crossings_rewritten}
        \crk(P_t; \col_t) &=& 2^{4t}\;\crk(P_0; \col_0) + \left[\sum\limits_{i=0}^{t-1} 16^{t-i-i} \left(\binom{2^i |P_0|}{2} - 2^i|P_0|\right)\right] \\
        &&+\ 4 \sum\limits_{p \in P_0} \sum\limits_{c = 1}^{k} \sum\limits_{d \in \{\ell,r\}} \left[\sum\limits_{i=0}^{t-1} 16^{t-i-i} \sum\limits_{j=1}^{2^i} \binom{f_{(o_1,o_2)}(\dset{c}{p},i,j)}{2}\right]\\
        &&+\ 2 \sum\limits_{p \in P_0} \sum\limits_{d \in \{\ell,r\}} \left[\sum\limits_{i=0}^{t-1} 16^{t-i-i} \sum\limits_{j=1}^{2^i} f_{(o_1,o_2)}(\dset{\cbar(p)}{p},i,j)\right],\\
\end{array}
\end{equation}

where $(o_1, o_2)$ depend on the current $p$, $c$ and $d$.
Let us denote the bracketed terms in \labelcref{fml:t_step_crossings_rewritten} by $A(|P_0|)$, $B_{(o_1,o_2)}(\dset{c}{p})$, and $C_{(o_1,o_2)}(\dset{c}{p})$, respectively.
Straightforward but long computations yield closed formulas for these eleven functions, each of which is of the form $2^{4t}p_4(x) + 2^{3t}p_3(x) + 2^{2t}p_2(x) + 2^{t}p_1(x)$ for polynomials $p_1, p_2, p_3, p_4$.
From the exact formulas, which can be found in \cref{app:closed_forms}, we obtain the following theorem.

\begin{theorem}
    \label{thm:nice_shape}
    Given a non-empty point set $P_0$, $|P_0| \geq 3$ and a $k$-edge-coloring $\col_0$, a matching $\hm_0$, and details at all vertices of $P_0$,
    there are $\alpha, \beta, \gamma, \delta \in \R, \alpha > 0, \beta < 0, \alpha + \beta + \gamma + \delta = \crk(P_0; \col_0)$ such that for any $t \in \N_0$
    $$\crk(P_t; \col_t) = \alpha 2^{4t} + \beta 2^{3t} + \gamma 2^{2t} + \delta 2^t.$$
\end{theorem}
\noindent\emph{Proof.}
    The existence of such $\alpha, \beta, \gamma, \delta \in \R$ is a direct consequence of \cref{fml:t_step_crossings_rewritten} being a finite sum over the functions $A(|P_0|)$, $B_{(o_1,o_2)}(\dset{c}{p})$, and $C_{(o_1,o_2)}(\dset{c}{p})$, each of which has the desired form.
	Then, $\alpha > 0$ and $\beta < 0$ follow from the signs of the relevant coefficients in the closed formulas when $|P_0|\geq 3$. Finally, setting $t = 0$ implies $\alpha + \beta + \gamma + \delta = \crk(P_0; \col_0)$.
\qed

\paragraph*{Asymptotics and Optimal Matchings.} With \cref{thm:nice_shape} we can now bound the \gcrkconstant by
\[ \gcrk \leq \lim\limits_{t \to \infty} \frac{\crk(P_t; \col_t)}{\binom{|P_t|}{4}} = \lim\limits_{t \to \infty} \frac{\alpha 2^{4t} + \mO(2^{3t})}{\frac{|P_0|^4}{24}2^{4t} + \mO(2^{3t})} = \frac{24 \alpha}{|P_0|^4}.\]

To compute $\alpha$, one only needs to determine $S^d_c(p)$ and $(o_1,o_2)$ for each $(p,c,d)$ and sum the contributions to $\alpha$ involving the $A,B$ and $C$-terms from \labelcref{fml:t_step_crossings_rewritten}.
We call the sum of these terms involving $S^d_c(p)$ for fixed $p$ the \emph{local $\alpha$ (at $p$)}.

Given $P_0, \col_0$ and $\hm_0$, the details that minimize the total~$\alpha$ result from choosing, at each $p\in P_0$, the details that minimize the local $\alpha$.
Due to this independence, even if only $P_0$ and $\col_0$ are given, the optimal matching can be found efficiently: Define weights $w$ on $P_0^2$ by setting $w(p,q)$ to the minimum local $\alpha$ at $p$ over all details at $p$ if $\hm_0(p)=q$ was fixed.
Further, let $H = (P_0 \cup \binom{P_0}{2}, \{(p, pq)\,\mid\,p,q \in P_0, p \neq q\})$ be the bipartite graph with edge weights $w_H(p,pq) = w(p,q)$.
An optimum matching which minimizes $\alpha$ in the doubling procedure corresponds to a $P_0$-saturating matching (in the usual sense) in $H$ with minimum weight, which can be determined in polynomial time~\cite{kuhnHungarianMethodAssignment1955}.

\section{Computational results for \texorpdfstring{$k \in \{2, \ldots 10\}$}{k in 2,..., 10}}

We focussed on obtaining upper bounds for the \gcrkconstant $\gcrk$ for $k = 2\ldots,10$.
For $k=2$, the authors of~\cite{aichholzer2ColoredCrossingNumber2019} provide a set of 135 points $\mathcal{P}'_2$ with a 2-edge-coloring $\col'_2$ such that $\cro_2(\mathcal{P}'_2; \col'_2) = 1470756$.
Using their doubling procedure via halving matchings on this instance, they obtain $\gcr_2 < 0.11798016$.
For the same instance, using our doubling procedure we obtain a better bound of $\gcr_2 < 0.11750015$ given by the optimum (non-halving) matching found by the bipartite-matching approach described above.

For $k \geq 3$ the best known upper bounds on $\gcrk$ in \labelcref{eq:crk_bounds} are from the book-crossing number.

To improve these bounds, we searched for point sets $\mathcal{P}_k$ and $k$-edge-colorings $\col_k$ for $k \geq 2$ with small $\crk(\mathcal{P}_k; \col_k)$.
To this end, we employed various \textsc{MAX-$k$-CUT} heuristics implemented in~\cite{rodriguesdesousaComputationalStudyBranching2022}.
Starting with $\mathcal{P}'_2$, we ran the heuristics on the crossing graph of the induced drawing to find a $k$-edge-coloring with few monochromatic crossings.
Keeping the coloring fixed, we further reduced the number of monochromatic crossing by perturbing the points.
Iterating these two steps, we obtained point sets $\mathcal{P}_k$ and $k$-edge-colorings $\col_k$ for $k=2,\ldots,10$.
Finding optimum matchings and details at every vertex via the bipartite matching approach, the upper bounds on $\gcrk$ detailed in \cref{tab:results} are obtained.
These improve the bounds from the $k$-page book crossing number by a factor of about 3.
The point sets $\mathcal{P}_k$ and their edge-colorings and matchings together with a Python script that determines the derived upper bound on $\gcrk$ can be found in~\cite{hahnInstancesLowGeometric2025}.
Let us note that while the optimum matchings for our instances are not halving, $\lset{c}{p}$ and $\rset{c}{p}$ tend to have similar values for most points $p$ and colors $c$.

\begin{table}[htb]
    \centering
	{\small{
    \sisetup{
        table-alignment-mode = format,
        table-number-alignment = center,
        table-format =1.8,
        round-precision=8,
        round-mode=places
    }
    \begin{tabular}{c | c | S | S | S | S[round-precision=3]}
        $k$ & $\crk(\mathcal{P}_k; \col_k)$ & {\text{UB on }$\gcrk$\text{ from }$\mathcal{P}_k$} & \text{UB on }$\gcrk$\text{ from~\cite{aichholzer2ColoredCrossingNumber2019}} & {\text{UB on }$\gcrk$\text{ from }$\bcrk$} & {\text{Improvement factor}}\\
        \hline
        2  & 1468394 & 0.11731412216972345   & 0.11798016 & 0.375        & \num{\fpeval{1/(0.11731412216972345   / 0.11798016  )}} \\
        3  &  732746 & 0.060624658294281826  & {-}        & 0.1851851852 & \num{\fpeval{1/(0.06108546988282801   / 0.1851851852)}} \\
        4  &  413342 & 0.03572150651315414   & {-}        & 0.109375     & \num{\fpeval{1/(0.03572150651315414   / 0.109375    )}} \\
        5  &  264459 & 0.023893255829769575  & {-}        & 0.072        & \num{\fpeval{1/(0.023893255829769575  / 0.072       )}} \\
        6  &  183248 & 0.017260492241176078  & {-}        & 0.0509259259 & \num{\fpeval{1/(0.017260492241176078  / 0.0509259259)}} \\
        7  &  133405 & 0.013140791599766348  & {-}        & 0.0379008746 & \num{\fpeval{1/(0.013140791599766348  / 0.0379008746)}} \\
        8  &   99638 & 0.010283341868079967  & {-}        & 0.0292968750 & \num{\fpeval{1/(0.010283341868079967  / 0.0292968750)}} \\
        9  &   78269 & 0.00845338638585603   & {-}        & 0.0233196159 & \num{\fpeval{1/(0.00845338638585603   / 0.0233196159)}} \\
        10 &   60922 & 0.0069267072517390045 & {-}        & 0.019        & \num{\fpeval{1/(0.0069267072517390045 / 0.019       )}} \\
    \end{tabular}
	}}
    \caption{Upper bounds on the \gcrkconstant for $k = 2, \ldots, 10$.}
    \label{tab:results}
\end{table}

\section{Conclusion}

We introduced a procedure that creates $k$-edge-colored straight-line drawings of $K_n$ for large $n$ with few monochromatic crossings, given an initial drawing for small $n$ with this property.
By finding $k$-edge-colored drawings with few monochromatic crossings, we kickstarted this procedure to improve the upper bounds on the \gcrkconstant for $k = 2, \ldots, 10$.
While our method is applicable for larger $k$, gaining on the upper bound from the book crossing number in~\cite{damianiUpperBoundCrossing1994,shahrokhiBookCrossingNumber1996} for all $k$ at once is still open.
We believe this is possible with a construction that does not arrange the points in convex position.

The lower bound on $\gcrk$ in \labelcref{eq:crk_bounds} is unlikely to be improved using the Crossing Lemma (unless a better one is found).
A more promising avenue could be the study of $\ell$-edges and $\leq \ell$-edges in a similar fashion as in~\cite{abrego$leqK$EdgesCrossings2011} and previously~\cite{aichholzerNewLowerBounds2007} for the (non-colored) rectilinear crossing constant.

\vspace*{1em}

\noindent\textbf{Acknowledgements.}
We thank the authors of~\cite{rodriguesdesousaComputationalStudyBranching2022} for providing us with their source code by making it publicly available at~\cite{rodriguesdesousaMax$k$CutBranchingAlgorithm2020} and Oswin Aichholzer for providing us with source code used for computational results in~\cite{aichholzer2ColoredCrossingNumber2019}.

\printbibliography

\newpage
\appendix

\section{Proof of \texorpdfstring{\cref{thm:t_step_crossings}}{Theorem 1}}
\label{app:proof_one_step_crossings}

\setcounter{theorem}{0}
\begin{theorem}
	Given a point set $P_0$, a $k$-edge-coloring $\col_0$, a matching $\hm_0$, and details at all vertices of $P_0$, the number of monochromatic crossings after $t$ iterations of the doubling construction is
    \begin{align}
    \begin{split}
        \crk(P_t; \col_t) =&\ 16^t\;\crk(P_0; \col_0) + \sum_{i=0}^{t-1}16^{t-i-1}\left[\binom{2^i|P_0|}{2} - 2^i|P_0|\right] \\
        &\ +\ 4 \sum_{p \in P_0} \sum_{c = 1}^{k} \sum_{i=0}^{t-1}16^{t-i-1}\sum_{j=1}^{2^i}\left[ \binom{\lset{c}{p^i_j}}{2} + \binom{\rset{c}{p^i_j}}{2} \right] \\
        &\ +\ 2 \sum_{p \in P_0} \sum_{i=0}^{t-1}16^{t-i-1} \sum_{j=1}^{2^i}\left[\lset{\cbar(p)}{p^i_j} + \rset{\cbar(p)}{p^i_j}. \right]
    \end{split}
    \end{align}
\end{theorem}

\begin{proof}
    We first proof the following claim about the number of crossings after a single iteration.
    \begin{claim}
        \label{claim:one_step_crossings}
        Given a point set $P_0$, a $k$-edge-coloring $\col_0$, a matching $\hm_0$, and details at all vertices of $P_0$, the number of monochromatic crossings after one step of the doubling construction is
        \begin{align*}
        \begin{split}
            \crk(P_1; \col_1) =&\ 16\;\crk(P_0; \col_0) +\ \binom{|P_0|}{2} - |P_0| \\
            &\ +\ 4 \sum_{p \in P_0} \sum_{c = 1}^{k} \left[\binom{\lset{c}{p}}{2} + \binom{\rset{c}{p}}{2}\right] \\
            &\ +\ 2 \sum_{p \in P_0} \left[\lset{\cbar(p)}{p} + \rset{\cbar(p)}{p}\right]
        \end{split}
        \end{align*}
    \end{claim}
    \begin{proof}[Proof of \cref{claim:one_step_crossings}]
        Each crossing in $P_1$ is determined by a quadruple of points in convex position whose two diagonals have the same color. Four such points can be the children of either 2, 3 or 4 points of $P_0$, see \cref{fig:proof_counting_crossings}. 
        We count crossings based on these three types (I, II, and III).
        \begin{enumerate}[I.]
            \item[I.] For each pair of points $p,q \in P_0$, the points $p_1, p_2, q_1, q_2$ form a monochromatic crossing in $P_1$ under $\col_1$ except if $pq$ is the matching edge of $p$ or $q$. As there are $|P_0|$ matching edges in total, we have $\binom{|P_0|}{2} - |P_0|$ such crossings.
            \item[II.] Given $p\in P_0$, we investigate under which conditions $p_1, p_2$ yield a monochromatic crossing in $P_1$ with two vertices $q_i, r_j$ of different parents $q,r \in P_1$. To this end, distinguish two further subcases.
            \begin{enumerate}
                \item[IIa.] If $\hm_0(p) \notin \{q, r\}$, then $p_1, p_2, q_i, r_j$ form a crossing if and only if $q$ and $r$ lie on the same side of $p \hm_0(p)$ in $P_0$ and $\col_0(p q) = \col_0(p r)$. Since this is independent of $i$ and $j$, there are $4 \sum_{c = 1}^{k} \binom{\lset{c}{p}}{2} + \binom{\rset{c}{p}}{2}$ crossings for any given $p$.
                
                \item[IIb.] Otherwise, w.l.og.\ $\hm_0(p) = q$ and then $p_1, p_2, q_i, r_j$ form a crossing if $q_i$ and $r$ lie on the same side of $p \hm_0(p)$ and $\col_0(p r) = \col_0(p q) = \cbar(p)$. Since for each $r$ with $\col_0(p r) = \cbar(p)$ exactly one of the $q_i$ lies on its side of $p \hm_0(p)$, and there are two possibilities for $j$, we have a total of $2 (\lset{\cbar(p)}{p} + \rset{\cbar(p)}{p})$ such crossings at $p$.
            \end{enumerate}
                \item[III.] Finally, four points in $P_1$ that are children of distinct points in $P_0$ form a monochromatic crossing if and only if their parents do. In total, this yields $16\;\crk(P_0; \col_0)$ crossings.
        \end{enumerate}
        The formula is obtained by summing over the different types of crossings.
    \end{proof}
    \begin{figure}[htb]
        \centering
        \captionsetup[subfigure]{labelformat=empty}
        \begin{subfigure}{0.45\textwidth}
            \centering
            \includegraphics[page=1]{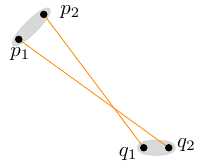}
            \caption{Type I}
        \end{subfigure}
        \begin{subfigure}{0.45\textwidth}
            \centering
            \includegraphics[page=2]{figures/crossing_counting_proof.pdf}
            \caption{Type IIa}
        \end{subfigure}
        \\
        \begin{subfigure}{0.45\textwidth}
            \centering
            \includegraphics[page=4]{figures/crossing_counting_proof.pdf}
            \caption{Type III}
        \end{subfigure}
        \begin{subfigure}{0.45\textwidth}
            \centering
            \includegraphics[page=3]{figures/crossing_counting_proof.pdf}
            \caption{Type IIb}
        \end{subfigure}
        \caption{The four cases we use to count the number of crossings after one step of the doubling construction.}
        \label{fig:proof_counting_crossings}
    \end{figure}

    By definition, subsequent iterations of the doubling construction behave analogously to the first and so the formula in \cref{claim:one_step_crossings} also holds if all mentions of $P_0$ and $P_1$ are replaced by $P_i$ and $P_{i+1}$ and in particular $P_{t-1}$ and $P_t$, yielding a formula for $\crk(P_t; \col_t)$.
    The claimed formula then follows by repeatedly expanding the $\crk(P_i;\col_i)$-terms for $i = t-1, t-2, \ldots, 1$.  
\end{proof}

\section{Proof of recurrence formulas}
\label{app:reccurence_formulas}

\begin{lemma}
    Let $T$ be an infinite full binary tree with children denoted left and right. Let further every node contain a real value $x^i_j$ where $i$ is the layer of the node (starting at 0) and $j \in \{1, \ldots, 2^i\}$ its position from left to right and let $x \coloneqq x^0_1$ be the value of the root node.
    If there exist values $o_1, o_2 \in \R$ such that for any $x^i_j$ the values of its children are
    $$x^{i+1}_{2j-1} = 2^ix + o_1 \qquad x^{i+1}_{2j} = 2^ix + o_2,$$
    then $x^i_j = 2^ix + o_1 (2^i - j) + o_2 (j - 1)$ for all $i \in \N_0, j \in \{1, \ldots 2^i\}$. In particular, we obtain the following closed formulas for $x^i_j$ for the indicated values of $o_1$ and $o_2$.
    \begin{table}[h]
        \centering
        {\tabulinesep=.2mm
        \begin{tabu}{c|c|c|c}
            \diagbox[height=1.5em]{$o_2$}{$o_1$} & 0 & 1 & 2\\
            \hline
            0 & $2^i x$       & $2^i x + 2^i - j$ &  - \\
            1 & $2^i x + j-1$ & $2^i x + 2^i - 1$ & $2^i x + 2^{i+1} - j - 1$ \\
        \end{tabu}}
    \end{table}
\end{lemma}
\begin{proof}
    We use induction on the layer $i$.
    On layer 0, there is only the value $x^0_1 = x$, for which the formula clearly holds.
    Now, let $i$ be arbitrary and assume $x^i_j = 2^ix + o_1 (2^i - j) + o_2 (j - 1)$ for all nodes on layer $i$. Let $x^{i+1}_{2j-1}$ and $x^{i+1}_{2j}$ be the values of two arbitrary sibling nodes on layer $i+1$. Then, by induction
    $$x^{i+1}_{2j-1} = 2x^i_j + o_1 = 2^{i+1}x + o_1(2^{i+1} - 2j - 1) + o_2 (2j - 2)$$
    and $$x^{i+1}_{2j} = 2x^i_j + o_2 = 2^{i+1}x + o_1 (2^{i+1} -2j) + o_2 (2j - 1),$$ as desired.
\end{proof}

\pagebreak
\section{Closed form solutions for \texorpdfstring{$A, B$ and $C$}{A, B and C}}
\label{app:closed_forms}

\allowdisplaybreaks
\begin{alignat*}{9}
    A(x) &=&\ &2^{4t}\left(\frac{x^2}{24} - \frac{3x}{28}\right) && &&-2^{2t}\frac{x^2}{24} &&+ 2^t\frac{3x}{28} \\
    B_{(0,0)}(x) &=&&2^{4t}\left(\frac{x^2}{16} - \frac{x}{24}\right) &&-2^{3t}\frac{x^2}{16} &&+ 2^{2t}\frac{x}{24} &&\\
    B_{(0,1)}(x) &=&&2^{4t}\left(\frac{x^2}{16} - \frac{x}{48} + \frac{1}{336}\right) &&-2^{3t}\left(\frac{x^2}{16} + \frac{x}{8} + \frac{1}{48}\right) &&+ 2^{2t}\left(\frac{x}{12} + \frac{1}{24}\right) &&- 2^t \frac{1}{42}\\
    B_{(1,0)}(x) &=&& B_{(0,1)}(x)&& && && \\
    B_{(1,1)}(x) &=&&2^{4t}\left(\frac{x^2}{16} + \frac{1}{112}\right) &&-2^{3t}\left(\frac{x^2}{16} + \frac{x}{8} + \frac{1}{16}\right) &&+ 2^{2t}\left(\frac{x}{8} + \frac{1}{8}\right) &&- 2^t \frac{1}{14} \\
    B_{(2,1)}(x) &=&&2^{4t}\left(\frac{x^2}{16} + \frac{x}{48} + \frac{3}{112}\right) &&-2^{3t}\left(\frac{x^2}{16} + \frac{3x}{16} + \frac{7}{48}\right) &&+ 2^{2t}\left(\frac{x}{6} + \frac{1}{4}\right) &&- 2^t \frac{11}{84}\\
    C_{(0,0)}(x) &=&&2^{4t}\frac{x}{12} && &&-2^{2t}\frac{x}{12} &&\\
    C_{(0,1)}(x) &=&&2^{4t}\left(\frac{x}{12} + \frac{1}{168}\right) && -2^{3t}\frac{1}{24}&&-2^{2t}\frac{x}{12} &&+ 2^t \frac{1}{28}\\
    C_{(1,0)}(x) &=&&C_{(0,1)}(x) && && &&\\
    C_{(1,1)}(x) &=&&2^{4t}\left(\frac{x}{12} + \frac{1}{84}\right) && &&-2^{2t}\left(\frac{x}{12} + \frac{1}{12}\right) &&+ 2^t \frac{1}{14}\\
    C_{(2,1)}(x) &=&&2^{4t}\left(\frac{x}{12} + \frac{1}{56}\right) && &&-2^{2t}\left(\frac{x}{12} + \frac{1}{8}\right) &&+ 2^t \frac{3}{28}\\
\end{alignat*}
\end{document}